\documentclass[twoside,leqno,twocolumn]{article}
\usepackage{ltexpprt}
\usepackage[cmex10]{amsmath}
\usepackage{amsfonts}
\usepackage{bm}
\usepackage{epsfig}
\usepackage{graphicx}
\usepackage{flushend}

\newtheorem{problem}{Problem}

\newcommand{\spara}[1]{\smallskip\noindent{\bf{#1}}}

\newcommand{\field}[1]{\ensuremath{\mathbb{#1}}}

\newcommand{\bff}{\ensuremath{\bm{f}}}
\newcommand{\bfz}{\ensuremath{\bm{z}}}

\newcommand{\bfs}{\ensuremath{\bm{s}}}

\newcommand{\eg}{\ensuremath{g}}

\newcommand{\NP}{\ensuremath{\mathbf{NP}}}
\newcommand{\bigO}{\ensuremath{{\cal O}}}

\newcommand{\campaign}{{\textsc{Campaign}}}
\newcommand{\internal}{{\textsc{i-Campaign}}}

\newcommand{\expressed}{{\textsc{Campaign}}}

\newcommand{\greedy}{{\texttt{Greedy}}}
\newcommand{\degree}{{\texttt{Degree}}}
\newcommand{\freedegree}{{\texttt{FreeDegree}}}
\newcommand{\rwr}{{\texttt{RWR}}}
\newcommand{\mins}{{\texttt{Min-S}}}
\newcommand{\minz}{{\texttt{Min-Z}}}
\newcommand{\icampaign}{{\texttt{i-Campaign}}}
\newcommand{\copynodes}{{\ensuremath{\overline{V}}}}
\newcommand{\copyedges}{{\ensuremath{\overline{E}}}}

\newcommand{\karate}{{{\bf karate club}}}
\newcommand{\lesmis}{{{\bf les miserables}}}
\newcommand{\dolphins}{{{\bf dolphins}}}

\newcommand{\bib}{{{\bf bibsonomy}}}
\newcommand{\dblp}{{{\bf dblp}}}

\begin{document}

\title{Opinion maximization in social networks}

\author{
Aristides Gionis\thanks{Aalto University, Finland. This work was done while the author was with Yahoo!\ Research.} \\ \and
Evimaria Terzi\thanks{Boston University, USA. Supported by NSF awards: CNS-1017529, III-1218437 and gifts from Microsoft and Google.} \\ \and
Panayiotis Tsaparas\thanks{University of Ioannina, Greece.}
}

\date{}

\maketitle

\begin{abstract}

The process of opinion formation through synthesis and contrast of different viewpoints
has been the subject of many studies in economics and social sciences.
Today, this process manifests itself also in online social networks and social media.
The key characteristic of successful promotion campaigns is that they take
into consideration such opinion-formation dynamics in order to create a
overall favorable opinion about a specific information item, such as a
person, a product, or an idea.

In this paper, we adopt a well-established model for social-opinion
dynamics and formalize the campaign-design problem
as the problem of identifying a set of target individuals whose
positive opinion about an information item
will maximize the overall positive opinion for the item in the social network.
We call this problem {\campaign}.
We study the complexity of the {\campaign} problem, and design algorithms for solving it.
Our experiments on real data demonstrate the efficiency and practical utility of our algorithms.
\end{abstract}

\section{Introduction}
Individuals who participate in social networks form their opinions through synthesis and contrast of different viewpoints they encounter in their social circles.
Such processes manifest themselves more strongly in online social networks and social media
where opinions and ideas propagate faster through virtual connections between social-network individuals.
Opinion dynamics have been of considerable interest to marketing and opinion-formation agencies, which are interested in raising public awareness on important issues (e.g., health, social justice), or increasing the popularity of person or an item (e.g., a presidential candidate, or a product).

Today, social-networking platforms and social media take up a significant amount of any promotion campaign budget.
It is not uncommon to see activist groups, political parties, or corporations launching campaigns primarily via Facebook or Twitter.
Such campaigns interfere with the opinion-formation process by influencing the opinions of appropriately selected individuals, such that the overall favorable opinion about the specific information item is strengthened.

The idea of leveraging social influence for marketing campaigns has been studied extensively in data mining.
Introduced by the work of Domingos and Richardson~\cite{richardson02mining}, and Kempe et al.~\cite{kempe2003maximizing} the problem of {\em influence maximization} asks to identify the most influential individuals, whose adoption of a product or an action will spread maximally in the social network.
This line of work employs probabilistic propagation models, such as the
{\em independent-cascade} or the {\em linear-threshold} model,
which specify how actions spread among the individuals of the social network.
At a high level, such propagation models distinguish the individuals to either {\em active} or {\em inactive}, and assume that active individuals influence their neighbors to become active according to certain probabilistic rules, so that activity spreads in the network in a {\em cascading} manner.

In this paper, we are interested in the process of how individuals form opinions, rather than how they adopt products or actions.
We find that models such as the independent cascade and the linear threshold are not appropriate for modeling the process of forming opinions in social networks.
First, opinions cannot be accurately modeled by binary states, such as being either active or inactive, but they can take continuous values in an interval.
Second, and perhaps more importantly, the formation of opinions does not resemble a discrete cascading process;
it better resembles a social game in which individuals constantly influence each other until convergence to an equilibrium.

Accordingly, we adopt the opinion-formation model of Friedkin and Johnsen~\cite{friedkin90social},
which assumes that each node $i$ of a social network $G=(V,E)$
has an internal opinion $s_i$ and an expressed opinion $z_i$.
While the internal opinions of individuals are fixed and not amenable
to external influences, the expressed opinions can change due to the influence
from their neighbors.
More specifically, people change their expressed opinion $z_i$ so that they minimize their \emph{social cost}, which is defined as the disagreement between their expressed opinion and the expressed opinion of their social neighbors, and the divergence between the expressed opinion and their true internal belief.
It can be shown that the process of computing expressed opinion values by repeated averaging
leads to a Nash equilibrium of the game where the utilities are the node social costs~\cite{DBLP:conf/focs/BindelKO11}.

Armed with this opinion-formation model, we introduce and study a problem similar to the influence-maximization problem, introduced by Kempe et al.~\cite{kempe2003maximizing}.
We consider a social network of individuals, each holding an opinion value about an information item.
We view an opinion as a numeric value between zero (negative) and one (positive).
We then ask to identify $k$ individuals,
so that once their expressed opinions becomes 1,
the opinions of the rest of the nodes --- at an equilibrium state --- are, on average, as positive as possible.
We call this problem {\campaign}.

 
Note that there is a fundamental qualitative difference between our framework and existing work on influence maximization: In contrast to other models used in influence-maximization literature,
our work views the nodes as rational individuals who wish to optimize their own objectives.
Thus, we assume that the opinions of individuals get formed through best-response dynamics of a social game, which in turn is inspired by classical opinion-dynamics models studied extensively in
economics~\cite{acemoglu11opinion,degrout74reaching,demarzo03persuasion,golub10naive,jackson08social}.

Furthermore, from a technical point of view, the maximization problem that results from our framework
requires a completely different toolbox of techniques
than the standard influence-maximization problem.
To this end, we exploit an interesting connection between our problem and absorbing random walks
in order to establish our technical results, including the complexity and the approximability of the problem.
Interestingly, as with the influence-maximization problem, we show that the objective function of the {\expressed} problem is submodular,  and thus, it can be approximated within factor $(1-\frac{1}{e})$
using a greedy algorithm.

In addition, motivated by the properties of the greedy algorithm, we propose scalable heuristics that can handle large datasets.
Our experiments show that in practice the heuristics have performance comparable to the greedy algorithm, and that are scalable to very large graphs.
Finally, we discuss two natural variants of the {\campaign} problem. Our discussion reveals a surprising property of the opinion formation process on \emph{undirected} graphs: the average opinion of the network depends only on the internal opinions of the individuals and not on the network structure.

\section{Related Work}
\label{section:related}
To the best of our knowledge, we are the first to formally
define and study the {\expressed} problem. However,
our work is related to a lot of existing work in economics,
sociology and computer science.

In his original work in 1974, DeGrout~\cite{degrout74reaching}
was the first to define a model where individuals organized in
a social network $G=(V,E)$ have a starting opinion (represented
by a real value) which they update
by repeatedly adopting the average opinion of their friends.
This model and its variants, has been subject of recent
studies in social sciences and economics~\cite{acemoglu11opinion,demarzo03persuasion,friedkin90social,
golub10naive,jackson08social}.
Most of these studies focus on identifying the conditions
under which such repeated-averaging models converge to a consensus.
In our work, we also adopt a repeated averaging model. However,
our focus is not the characterization or the reachability of consensus. 
In fact, we do not assume consensus is ever reached.
Rather, we assume that the individuals participating in the network
reach a stable point, where everyone has crystallized
a personal opinion. As has been recently noted by social scientist
Davide Krackhardt~\cite{krackhardt09plunge}, studies of such non-consensus
states are much more realistic since consensus states are rarely reached.

Key to our paper is the work by Bindel et
al.~\cite{DBLP:conf/focs/BindelKO11}. In fact, we adopt
the same opinion-dynamics model as in~\cite{DBLP:conf/focs/BindelKO11},
where individuals selfishly form opinions to minimize their personal cost.
However, Bindel et al.\ focus
on quantifying the social cost of this lack of central coordination between individuals, i.e.,
the \emph{price of anarchy}, and they consider a network-design problem
with the objective of reducing the social cost at equilibrium.
Our work on the other hand,
focuses on designing a promotion campaign
so that, at equilibrium, the overall positive inclination
towards a particular opinion is maximized.

Recently, there has been a lot of work
in the computer-science literature
on identifying a set of individuals
to advertise an item (e.g., a product or an idea)
so that the spread of the item in the network is maximized.
Different assumptions about how information items
propagate in the network has led to a rich literature
of targeted-advertisement methods (e.g., see ~\cite{chen08approximability,evendar2007note,kempe2003maximizing,richardson02mining}).
Although at a high level our work has the same goal as all of these methods,
there are also important differences, as we have already discussed in the introduction.

%

\section{Problem definition}
\label{section:notation}
\subsection{Preliminaries.}

We consider a social graph $G=(V,E)$ with $n$ nodes and $m$ edges.
The nodes of the graph represent people and the edges represent social affinity between them.
We refer to the members of the social graph by letters such as $i$ and~$j$, and we write $(i,j)$ to denote the edges of the graph.
With each edge $(i,j)$ we associate a weight $w_{ij}\ge 0$,
which expresses the strength of the social affinity or influence from person~$i$ to person~$j$.
We write $N(i)$ to denote the {\em social neighborhood} of person~$i$, that is,
$N(i) = \{ j \mid (i,j)\in E\}$.
Unless explicitly mentioned, we do not make an assumption whether the graph $G$ is directed or undirected; most of our results and our algorithms carry over for both types of graphs.
The directed-graph model is more natural as in many real-world situations the influence $w_{ij}$ from person $i$ to person $j$ is not equal to~$w_{ji}$.

Following the framework
of Bindel et al.~\cite{DBLP:conf/focs/BindelKO11}
we assume that person $i$ has a {\em persistent internal opinion} $s_i$, which remains unchanged from external influences.
Person $i$ has also an {\em expressed opinion} $z_i$, which depends on their internal opinion $s_i$,
as well as on the expressed opinions of their social neighborhood $N(i)$.
The underlying assumption is that
individuals form opinions that combine their internal predisposition with the opinions of those in their social circle.

We model the internal and external opinions $s_i$ and $z_i$ as real values in the interval~$[0,1]$.
The convention is that 0 denotes a negative opinion, and 1 a positive opinion.
The values in-between capture different shades of positive and negative opinions.
Given a set of expressed opinion values for all the people in the social graph, represented by an {\em opinion vector} $\bfz= (z_i : i\in V)$,
and the vector of internal opinions $\bfs = (s_i: i \in V)$,
we consider that the \emph{personal cost} for individual $i$ is
\begin{equation}
\label{equation:nash}	
c_i(\bfz) = (s_i-z_i)^2 + \sum_{j\in N(i)} w_{ij}(z_i - z_j)^2.
\end{equation}
This cost models the fact that the expressed opinion $z_i$ of an individual $i$ is a ``compromise'' between their own internal belief $s_i$ and the opinions of their neighbors.
As the individual $i$ forms an opinion $z_i$, their internal opinion $s_i$ and the opinions $z_j$ of their neighbors $j\in N(i)$ may have different importance.
The relative importance of those opinions is captured by the weights~$w_{ij}$.

Now assume that, as a result of social influence and conflict resolution, every individual $i$ is selfishly minimizing their own cost $c_i(\bfz)$.
If the internal opinions are \emph{persistent} and cannot change (an assumption that we carry throughout),
minimizing the cost $c_i$ implies changing the expressed opinion $z_i$
to the weighted average among the internal opinion $s_i$ and the expressed opinions of the neighbors of~$i$.
In other words,
\begin{equation}
\label{equation:averaging}
z_i = \frac{s_i+\sum_{j\in N(i)}w_{ij}z_j}{1+\sum_{j\in N(i)}w_{ij}}.
\end{equation}
In fact, it can be shown that if every person $i$ iteratively
updates their expressed opinion using Equation~(\ref{equation:averaging}),
then the iterations converge to a unique \emph{Nash Equilibrium}
for the game with utilities of players expressed by Equation~\eqref{equation:nash}.
That is, the stationary vector of opinions $\bfz$ is such that
no node $i$ has an incentive to change their opinion to improve their cost $c_i$.

\subsection{Problem definition.}

The goal of a promotion campaign is to improve the overall opinion about a product, person, or idea in a social network.
Given an opinion vector $\bfz$, we define the \emph{overall opinion} $g(\bfz)$ as
\[
g(\bfz)=\sum_{i=1}^nz_i,
\]
which is also proportional to the \emph{average expressed opinion} of the individuals in $G$.
The goal of a campaign is to maximize $g(\bfz)$.
%
Following the paradigm of Kempe et al.~\cite{kempe2003maximizing}, we
assume that such a campaign relies on selecting
a set of \emph{target nodes} $T$, which are going to be convinced to change
their expressed opinions to $1$.
For the rest of the discussion, we will use
$g(\bfz\mid T)$ to denote the overall opinion in the network, when vector $\bfz$
is the Nash-equilibrium vector obtained under the constraint that
the expressed opinions of all nodes in $T$ are
fixed to 1.
Given this notation, we can define the {\expressed} problem as follows.

\begin{problem}[{\expressed}]\label{problem:expressed}
Given a graph $G=(V,E)$ 
and an integer $k$, identify a set $T$ of $k$ nodes such that
fixing the expressed opinions of the nodes in $T$ to $1$, maximizes the overall opinion $g(\bfz\mid T)$.
\end{problem}

We emphasize that fixing $z_i=1$ for all $i\in T$
means that Equation~\eqref{equation:averaging}
is only applied for the $z_j$'s such that $j\notin T$, while
for the nodes $i\in T$ the values $z_i$ remain $1$.

The definition of the {\expressed} problem reflects our belief of what constitutes a feasible and effective campaign strategy.
Expressed opinions are more amenable to change, and have stronger effect on the overall opinion in the social network. Thus, it is reasonable for a campaign to target these opinions. We note that other campaign strategies are also possible, resulting in different problem definitions.
For example, one can define the problem where the campaign aims at changing the fundamental beliefs of people by altering their internal opinions $s_i$.
It is also conceivable to ask whether it is possible to improve the overall opinion $g(\bfz)$ by introducing a number of new edges in the social graph, e.g., via a link-suggestion application.
We discuss both of these variants at the end of the paper.
It turns out that from the algorithmic point of view, both these problems  are relatively simple.
For instance, we can show that for undirected graphs, surprisingly, it is not possible to improve the overall opinion $g(\bfz)$ by introducing new edges in the graph.


\subsection{Background.}
\label{sec:background}

We now show the connection between computing the Nash-equilibrium opinion vector $\bfz$ and a random walk on a graph with absorbing nodes.
This connection is essential in the analysis of the {\campaign} problem.

\spara{Absorbing random walks:}
Let $H=(X,R)$ be a graph with a set of $N$ nodes $X$, and a set of edges $R$.
The graph is also associated with the following three $N\times N$ matrices:
($i$) the \emph{weight matrix}
$W$ with entries $W(i,j)$ denoting the weight of the edges;
($ii$) the \emph{degree matrix} $D$, which is a diagonal matrix such that $D(i,i)=\sum_{j=1}^{N}W(i,j)$;
($iii$) the \emph{transition matrix} $P=D^{-1}W$,
which is a row-stochastic matrix; $P(i,j)$ expresses the probability of moving from node $i$ to node $j$ in
a random walk on the graph $H$.

In such a random walk on the graph $H$, we say that a node $b \in X$ is an \emph{absorbing node},
if the random walk can only transition into that node, but not out of it
(and thus, the random walk is absorbed in node $b$).
Let $B \subseteq X$ denote the set of all absorbing nodes of the random walk.
The set of the remaining nodes $U = X \setminus B$ are non-absorbing, or
\emph{transient} nodes.
Given this partition of the states in $X$,
the transition matrix of this random walk can be written as follows:
\[
P = \begin{pmatrix} P_{UB} & P_{UU} \\ I & O \end{pmatrix}.
\]
In the above equation, $I$ is an $(N-|U|)\times (N-|U|)$ identity matrix and $O$
a matrix with all its entries equal to $0$;
$P_{UU}$ is the $|U|\times |U|$ sub-matrix of $P$ with the transition probabilities between transient states; and
$P_{UB}$ is the $|U|\times |B|$ sub-matrix of $P$ with the transition probabilities from transient to absorbing states.

An important quantity of an absorbing random walk is the
expected number of visits to a transient state $j$
when starting from a transient state $i$ before being absorbed.
The probability of transitioning from $i$ to $j$ in exactly $\ell$
steps is the $(i,j)$-entry of the matrix $\left(P_{UU}\right)^\ell$.
Therefore, the probability that a random walk starting from state~$i$ ends in $j$
without being absorbed is given by the $(i,j)$ entry of the  $|U|\times |U|$ matrix
\[
F = \sum_{\ell=0}^{\infty}\left(P_{UU}\right)^{\ell}=\left(1-P_{UU}\right)^{-1},
\]
which is known as the \emph{fundamental matrix} of the absorbing random walk.
Finally, the matrix
\[Q_{UB} = F\,P_{UB}\]
is an $|U|\times |B|$ matrix, with $Q_{UB}(i,j)$ being the probability
that a random walk which starts at transient state $i$ ends up being absorbed
at state $j\in B$.

Assume that each
absorbing node $j \in B$ is associated with a fixed value $b_j$.
If a random walk starting from transient node $i\in U$ gets absorbed in an absorbing node $j\in B$, then we
assign to node $i$ the value $b_j$.
The probability of  the random walk starting from node $i$
to be absorbed in $j$ is $Q_{UB}(i,j)$. Therefore, the expected value of $i$ is $f_i = \sum_{j\in B} Q_{UB}(i,j) b_j$.
If $\bff_U$ is the vector with the expected values for all $i\in U$, and
$\bff_B$ keeps the values $\{b_j\}$ for all $j\in B$, then we have that
\begin{equation}\label{equation:equilibrium}
\bff_U = Q_{UB}\bff_B.
\end{equation}

A fundamental observation, which highlights the connection between our work and random walks with absorbing states, is that the expected value $f_i$ of node $i\in U$ can be computed by repeatedly averaging the values of the neighbors of $i$ in the graph~$H$.
Therefore, the computation of the Nash-Equilibrium opinion vector $\bfz$ can be done using
Equation~\eqref{equation:equilibrium} on an appropriately constructed graph~$H$.
We discuss the construction of $H$ below. More details on absorbing walks can be found in the excellent monograph of Doyle and
Snell~\cite{doyle1984random}.

\spara{The augmented graph.}
We will now show how the theory of absorbing random walks described above can be leveraged
for solving the {\expressed} problem.
This connection is achieved by performing a random walk with absorbing states on an \emph{augmented graph} $H=(X,R)$, whose construction we describe below.

Given a social network $G=(V,E)$ where every edge $(i,j)\in E$ is associated
with weight $w_{ij}$, we construct the augmented graph $H=(X,R)$ of $G$ as follows:
\begin{itemize}
\item[(i)]
the set of vertices $X$ of $H$ is defined as $X = V\cup {\copynodes}$,
where ${\copynodes}$ is a set of $n$ new nodes such that for each node $i \in V$
there is a copy $\sigma(i)\in\copynodes$;
\item[(ii)]
the set of edges $R$ of $H$ includes all the edges $E$ of $G$, plus a new set of edges between each node $i\in V$ and its copy $\sigma(i)\in\copynodes$.
That is, $R=E\cup {\copyedges}$, and ${\copyedges}=\{ (i,\sigma(i)) \mid  i\in V \}$;
\item[(iii)]
the weights of all the new edges $(i,\sigma(i))\in R$ are set to 1, i.e.,
$W(i,\sigma(i))=1$. For $i,j\in V$, the weight of the edge
$(i,j)\in R$ is equal to the weight of the corresponding edge in $G$, i.e., $W(i,j)=w_{ij}$.
\end{itemize}

Our main observation is that we can compute the opinion vector $\bfz$ that corresponds to the Nash equi\-lib\-rium defined by Equation~\eqref{equation:nash} by performing an  absorbing random walk on the graph $H$.
In this random walk, we set $B={\copynodes}$ and $U=V$, that is, we make all copy nodes in ${\copynodes}$ to be absorbing.
We also set $\bff_B = \bfs$, that is, we assign value $s_i$ to each absorbing node $\sigma(i)$. The Nash-equilibrium opinion-vector $\bfz$
can be computed using Equation~\eqref{equation:equilibrium}, that is, $\bfz = Q_{UB}\bfs$.
The opinion $z_i =\sum_{j\in B} Q_{UB}(i,j) s_j$ is the expected internal opinion value at the node of absorption for a random walk that starts from node $i\in V$.
Given the vector $\bfz$ we can compute the overall opinion $g(\bfz)$.

The {\expressed} problem can be naturally defined in this setting.
Selecting a set of nodes $T$ is equivalent to adding the nodes in $T$ into the set of absorbing nodes $B$, and assigning them value 1.
That is, we have $B={\copynodes}\cup T$ and $U=V\setminus T$.
For the vector $\bff_B$, we have $f_{\sigma(i)} = s_i$ for all $\sigma(i) \in {\copynodes}$, and $f_j = 1$ for all $j \in T$.
We use Equation~\eqref{equation:equilibrium} to compute vector $\bfz$
and using this $\bfz$, we can then compute the overall opinion $g(\bfz \mid T)$.
Hence, the {\expressed} problem becomes the problem of selecting a set of $k$ nodes $T \subseteq V$ to make absorbing with value 1, such that $g(\bfz \mid T)$ is maximized.

\section{Problem complexity}
\label{section:complexity}

In this section, we establish the complexity of the {\expressed} problem by showing that it is an {\NP}-hard problem.
We also discuss properties of the objective function $g(\bfz\mid T)$, which give rise
to a constant-factor approximation algorithm for the {\expressed} problem.

\begin{theorem}
\label{theorem:NP}	
Problem \expressed\ is \NP-hard. 	
\end{theorem}

The proof of the theorem appears in the Appendix~\ref{appendix:np-hard}.
The proof relies on a reduction from the {\sc Vertex Cover on Regular Graphs} problem (VCRG)~\cite{feige03vertex}.

Since the {\expressed} problem is {\NP}-hard,
we are content with algorithms that approximate the optimal solution in polynomial time.
Fortunately, we can show that the function $\eg(\bfz \mid T)$ is monotone and submodular,
and thus a simple greedy heuristic yields a constant-factor approximation to the optimal solution.



\begin{theorem}
\label{theorem:submodular}	
The function $\eg(\bfz\mid T)$ is monotone and sub\-modular.
\end{theorem}

\begin{proof}
We only give here a proof sketch.
A detailed proof is given in Appendix~\ref{appendix:submodularity}.
Recall that $\eg(\bfz\mid T) = \sum_{i\in V} z_i$.
In the absorbing random walk interpretation of the opinion formation process,
we have shown that the expressed opinion of node $i$ is the expected opinion value at the
point of absorbtion for a random walk that starts from node $i$.
That is, $z_i = \sum_{b \in B}P(b\mid i)f_b$, where $B$ is the set of
absorbing nodes, $P(b\mid i)$ is the probability of the random walk starting from node $i$
to be absorbed at node $b$, and $f_b$ the opinion value at node $b$.
When we add a node $x$ to $T$, and hence to the set $B$, some of the probability mass
of the random walk will be absorbed at $x$. Since $x$ has the maximum possible opinion
value, $f_x = 1$, if follows that $z_i$ can only increase, and thus $\eg(\bfz\mid T)$ is monotone.
Furthermore, the less competition there is for $x$ (i.e., the smaller the size of $B$),
the more mass of the random walk will be absorbed in $x$, and the larger the increase of $\eg(\bfz\mid T)$.
Hence $\eg(\bfz\mid T)$ is submodular.
\end{proof}


\section{Algorithms}
\label{section:algorithms}

\subsection{Estimating the Nash-equilibrium vector \bfz.}
\label{sec:opinion-estimation}

A central component in all the algorithms presented in this section is the estimation of the
the opinion function $g(\bfz\mid T)$.
In Section~\ref{sec:background}, we have already discussed that
this can be done by evaluating Equation~\eqref{equation:equilibrium}.
For appropriately defined sets $U$ and $B$,
this requires
computing the matrix
$Q_{{UB}} = \left(1-P_{{UU}}\right)^{-1}P_{{UB}}$.
Hence, this calculation involves a matrix inversion, which is very inefficient.
The reason is that despite the fact that the social graph is typically sparse,
matrix inversion does not preserve sparseness.
Thus, it may be too expensive to even store the matrix~$Q_{{UB}}$.

Instead, we resort to the power-iteration method implied by Equation~\eqref{equation:averaging}:
at each iteration we update the opinion $z_i$ of a node $i\in V$ by averaging
the opinions of its neighbors $j\in N(i)$ and its own internal opinion $s_i$.
During the iterations we do not update the values of opinions that are fixed.
This power-iteration method is known to converge to the equilibrium vector $\bfz$,
and it is highly scalable, since it only involves multiplication of a sparse matrix with a vector.
For a graph with $n$ nodes, $m$ edges, and thus, average degree $d=\frac{2m}{n}$, the algorithm requires $\bigO(nd)=\bigO(m)$ operations per iterations.
Therefore, the overall running time is $\bigO(mI)$, where $I$ is the total number of iterations.
In our experiments we found the the method converges in around 50-100 iterations, depending on the dataset.

\subsection{Algorithms for the {\expressed} problem.}
\label{sec:algos2}

Our algorithms for the \expressed\ problem, include
a constant-factor approximation algorithm as well as several efficient and
effective heuristics.

\spara{The {\greedy} algorithm.}
It is known that the greedy algorithm is a $(1-\frac 1e)$-approximation
algorithm for maximizing a submodular function
$h:Y\rightarrow \field{R}$ subject to cardinality constraints, i.e.,
finding a set $A \subseteq Y$ that maximizes $h(A)$ such that $|A|\le k$~\cite{nemhauser78analysis}.
Consequently, the {\greedy} algorithm
constructs the set $T$ by adding one node in each iteration.
In the $t$-th iteration the algorithm extends the set $T^{(t-1)}$ by adding the node $i$
that maximizes $\eg(\bfz\mid T^{(t)})$ when setting $z_i=1$.


\begin{figure*}[t]
\centering
\includegraphics[width=0.98\textwidth]{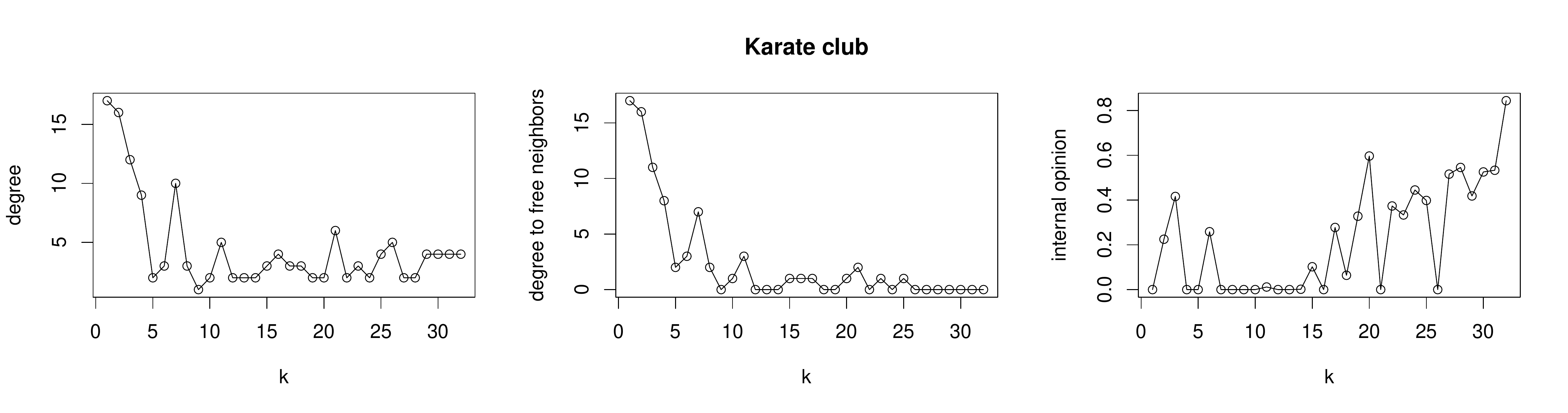}
\caption{\label{figure:newman-order}
Measures of graph nodes plotted in order selected by the \greedy\ algorithm.}
\end{figure*}

The computational bottleneck of {\greedy} is
due to the fact that
we need to compute the Nash-equilibrium opinion vector \bfz\
that results from setting $z_t=1$ for all $t\in T\cup\{j\}$, and we need to do such a computation for all candidate nodes $j\in V\setminus T$.
Overall, for a solution $T$ of size $|T|=k$ {\greedy} needs to perform $\bigO(nk)$ computations of finding the optimal vector~\bfz.
As we saw, each of these computations is performed by a power-iteration in time $\bigO(mI)$, yielding an overall running time $\bigO(nmkI)$.
Such a running time is super-quadratic and therefore the algorithm is not scalable to very large datasets.

One way to speedup the algorithm is by storing, for each node that it is not yet selected, its marginal improvement on the score, at the last time it was computed.
This speedup, which is commonly used in optimization problems with submodular functions~\cite{leskovec07cost-effective},
is not adequate to make the greedy algorithm applicable for large data,
at least for the version of the algorithm described here.
The reason is that in the very first iteration there is no pruning and therefore we need to make $O(n)$ power-iteration computations, yielding again a quadratic algorithm.
To overcome these scalability limitations, we present a number of scalable heuristics.

\subsection{Designing the heuristics.}

To characterize the nodes selected by \greedy\ we execute the algorithm on small datasets, and we compute a number of measures for each node selected by \greedy.
In particular, for each node we compute measures such as its degree, the average degree of its neighbors, the maximum degree of its neighbors, the value of its internal opinion $s_i$, the average value of $s_i$ over its neighbors, and so on.
Three of the features with the most clear signal are shown in Figure~\ref{figure:newman-order} for the {\karate} dataset (described in detail in Section~\ref{section:experimental-evaluation}).
We obtain similar behavior on all the datasets we tried.

In Figure~\ref{figure:newman-order} we plot measures of nodes in the order selected by the \greedy.
A good measure would be one that is monotonic with respect to this order.
In the first panel, we show the degree of a node in the selection order of \greedy,
and we see that \greedy\ tends to select first high degree nodes.
As shown in the second panel, this dependence is even more clear
for the  {\em free degree}, i.e., the number of neighbors that are not already selected by {\greedy}.
Finally, in the third panel of Figure~\ref{figure:newman-order} we see the internal opinion $s_i$ of nodes in the order selected by the \greedy.
We see that the \greedy\ tends to select first nodes with low internal opinion.
There are a few exceptions of nodes with high internal opinion $s_i$ selected at the initial steps of greedy.
Such nodes are nodes with high degree, connected to many nodes with small values of~$s_i$.

Armed with intuition from this analysis we now proceed to describe our heuristics.

\spara{The \degree\ algorithm.}
This algorithm simply sorts the nodes of $G=(V,E)$ in decreasing order of their
in degree and forms the set of target nodes $T$ by picking the top-$k$ nodes of the ranking.
The running time of {\degree} is $\bigO(n\log n)$, i.e., the time required for sorting.

\spara{The \freedegree\ algorithm.}
This algorithm is a ``greedy'' variant of {\degree};
{\freedegree} forms the set $T$ iteratively by choosing at every iteration the node with the highest free degree.
The free degree of a node is the sum of the weights of the edges that are incident to it and
are are connected to nodes not already in $T$.
When the set $T$ consists of $k$ nodes, the running time of {\freedegree} is $\bigO(kn)$.

\spara{The \rwr\ algorithm.}
As we saw in Figure~\ref{figure:newman-order}, a good choice for nodes to be added in the solution are not only the nodes of high degree but also the nodes of small value of internal opinion~$s_i$.
The \rwr\ algorithm combines both of these features: selecting nodes with high degree and with small~$s_i$.
This is done by performing a {\underline r}andom {\underline w}alk with {\underline r}estart (\rwr),
where the probability of restarting at a node~$i$ is proportional to $r_i=s_{\max}-s_i$,
where $s_{\max}=\max_{i\in V}s_i$, and ordering the nodes according to the resulting stationary distribution.
The intuition is that a random walk favors high-degree nodes, and using the specific restart probabilities favors nodes with low value of~$s_i$.


For the restart probability, we use the parameter $\alpha=0.15$, which has been established as a standard parameter of the PageRank algorithm~\cite{brin98anatomy}.
Making one \rwr\ computation can be achieved by the power-iteration method, which similarly to computing the optimal vector~\bfz, has running time $\bigO(mI)$.
Therefore, the overall running time of the algorithm for selecting a set $T$ of size $k$ is $\bigO(mkI)$.

\spara{The {\mins} and {\minz} algorithms.}
The {\mins} algorithm simply selects the $k$ nodes with the smallest value $s_i$.
This heuristic is motivated by the observation that the {\greedy} algorithm tends to select nodes with small value $s_i$.
For completeness, we also experiment with the {\minz} algorithm, which greedily selects and add in the solution set $T$ the node that at the current iteration has the smallest value of expressed opinion~$z_i$.

\section{Experimental evaluation}
\label{section:experimental-evaluation}

\begin{figure*}[t]
\centering
\includegraphics[width=0.3\textwidth]{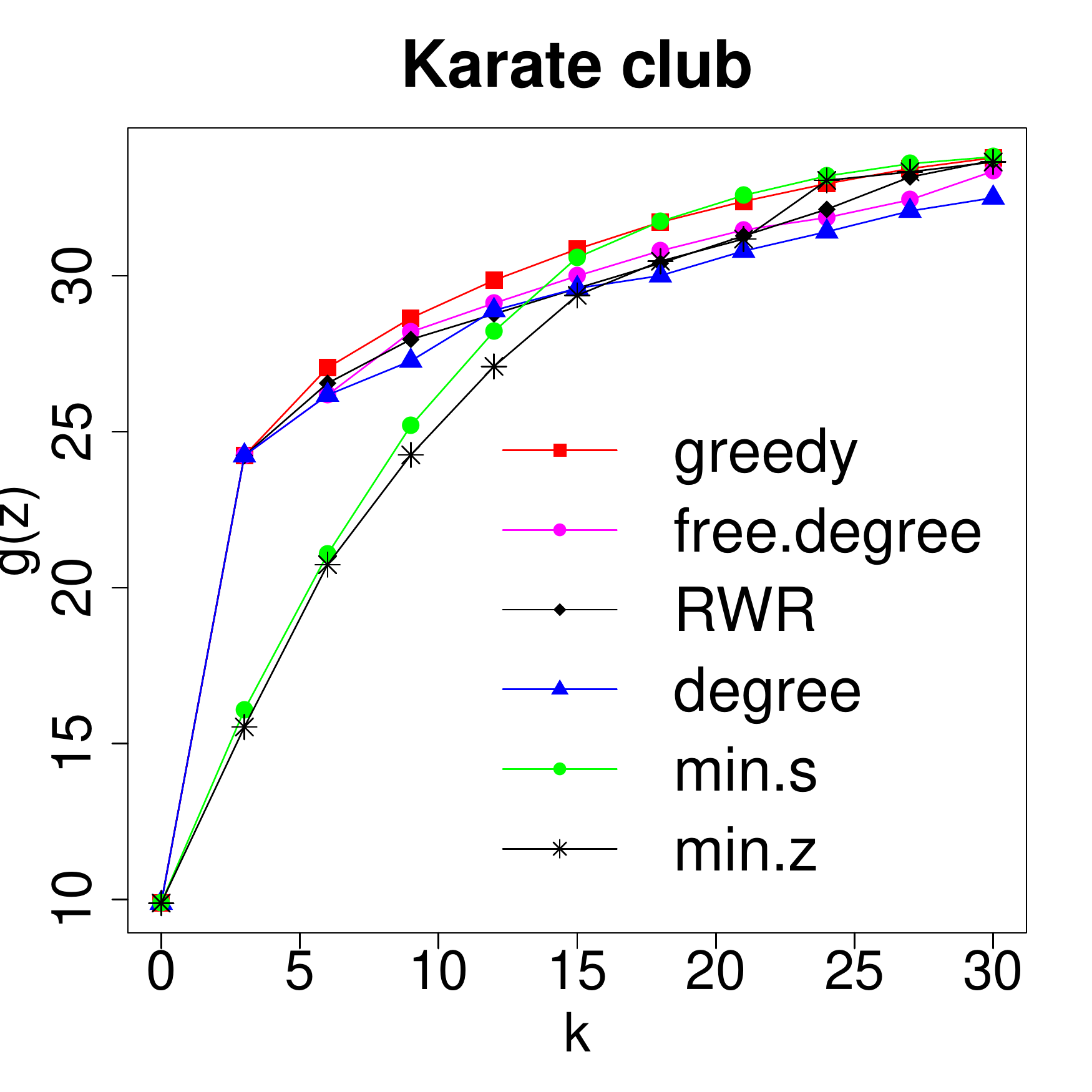}
\hspace{0.5cm}
\includegraphics[width=0.3\textwidth]{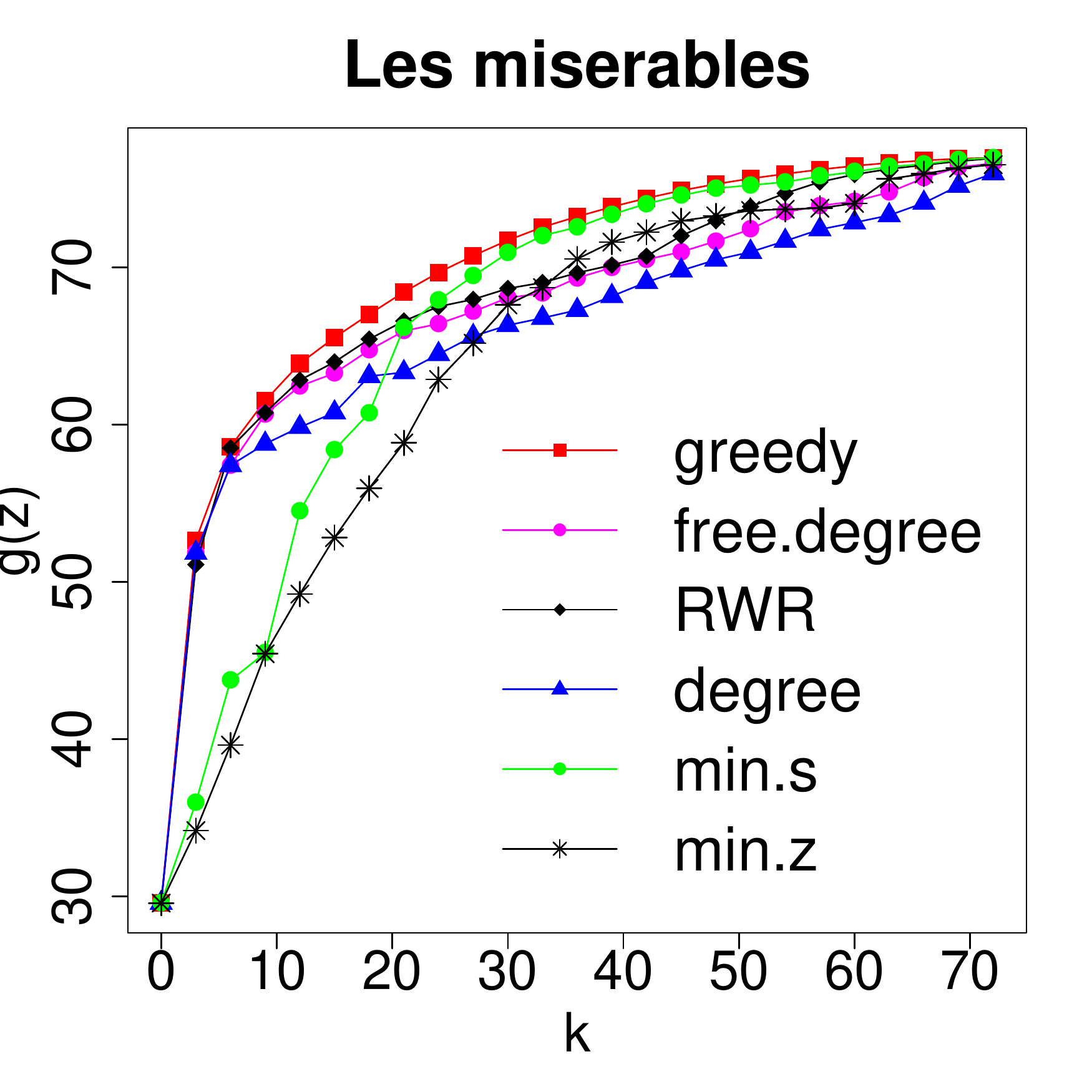}
\includegraphics[width=0.3\textwidth]{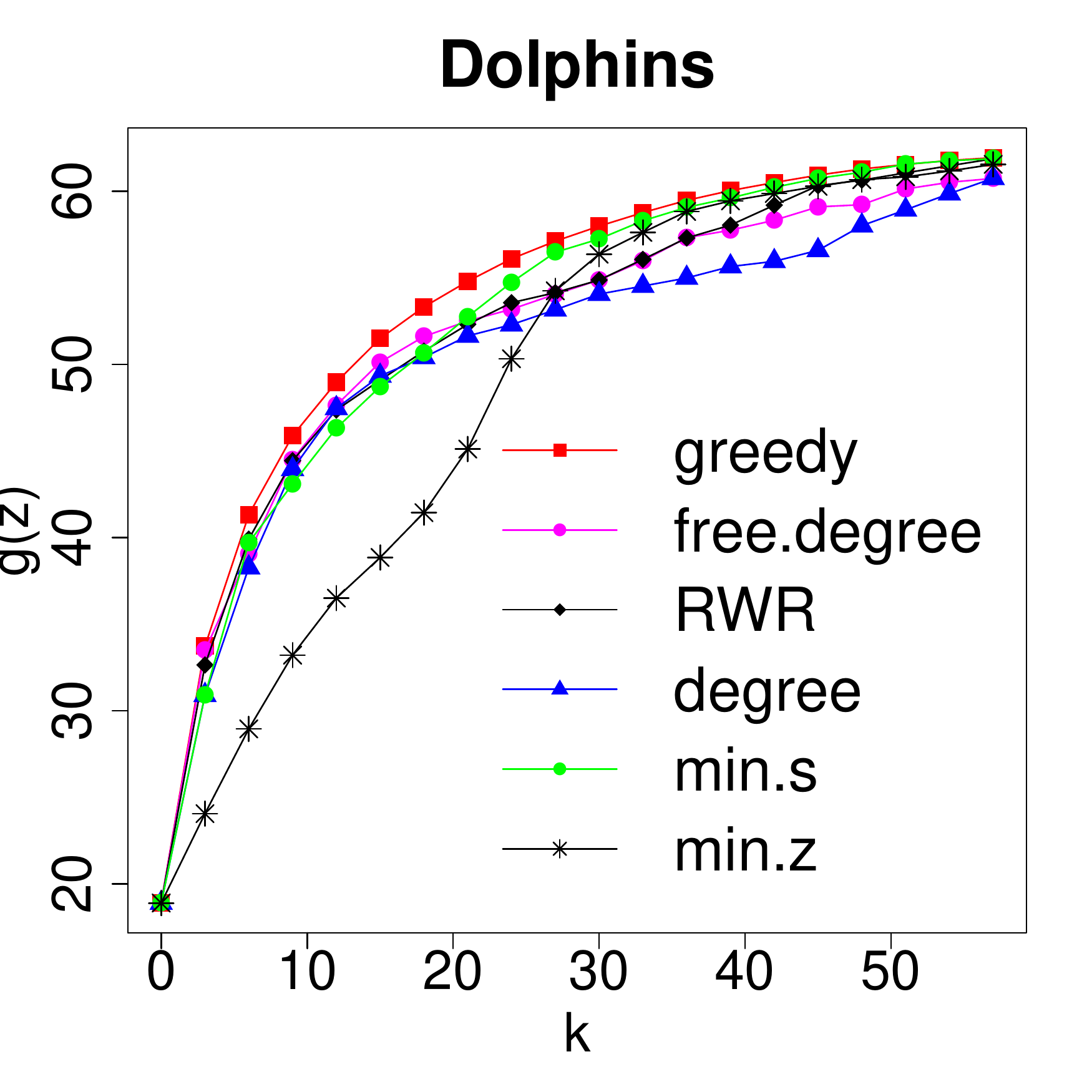}
\caption{\label{figure:newman}
Performance of the algorithms on small networks.}
\end{figure*}

\begin{figure*}[t]
\centering
\includegraphics[width=0.47\textwidth]{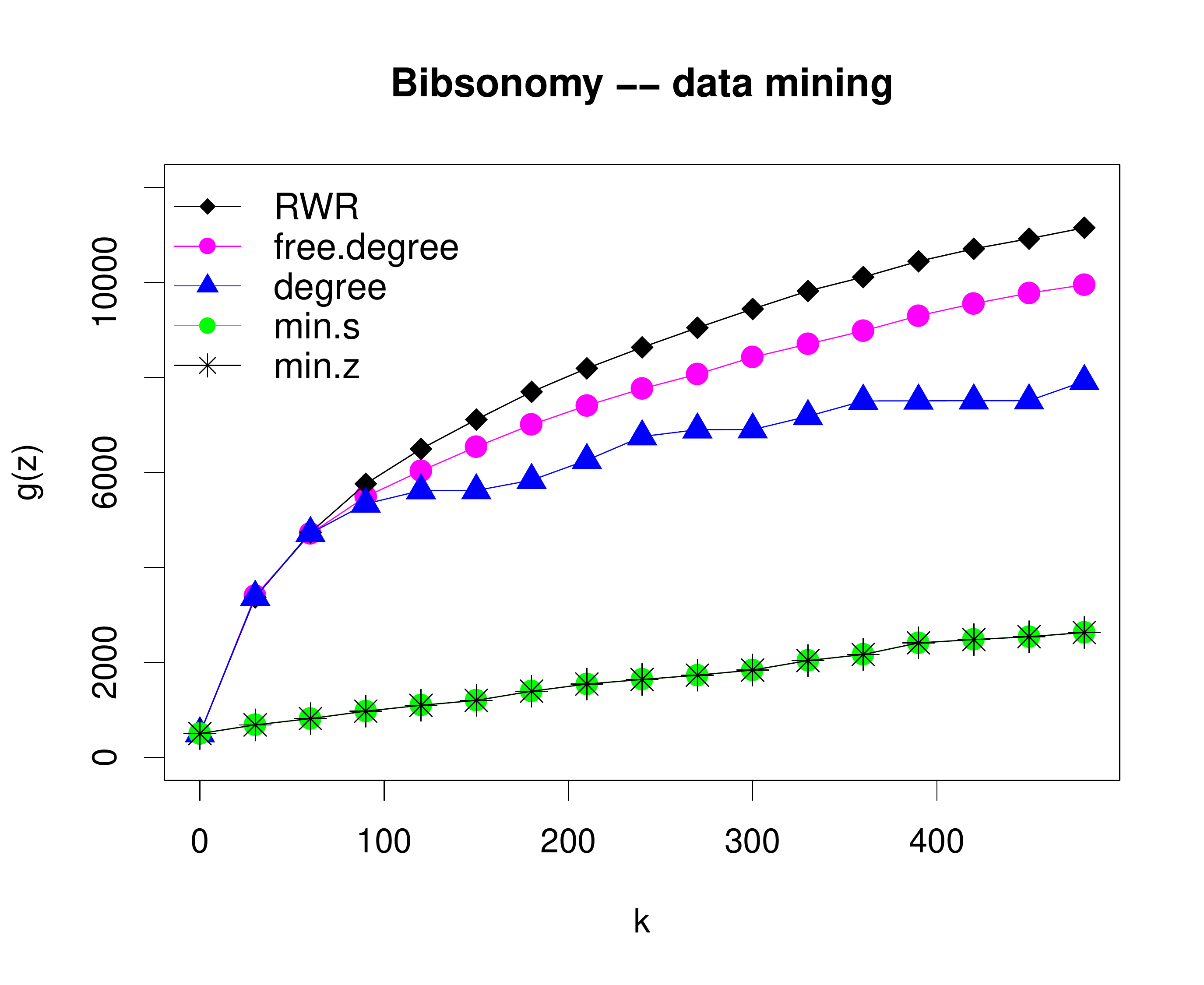}
\includegraphics[width=0.47\textwidth]{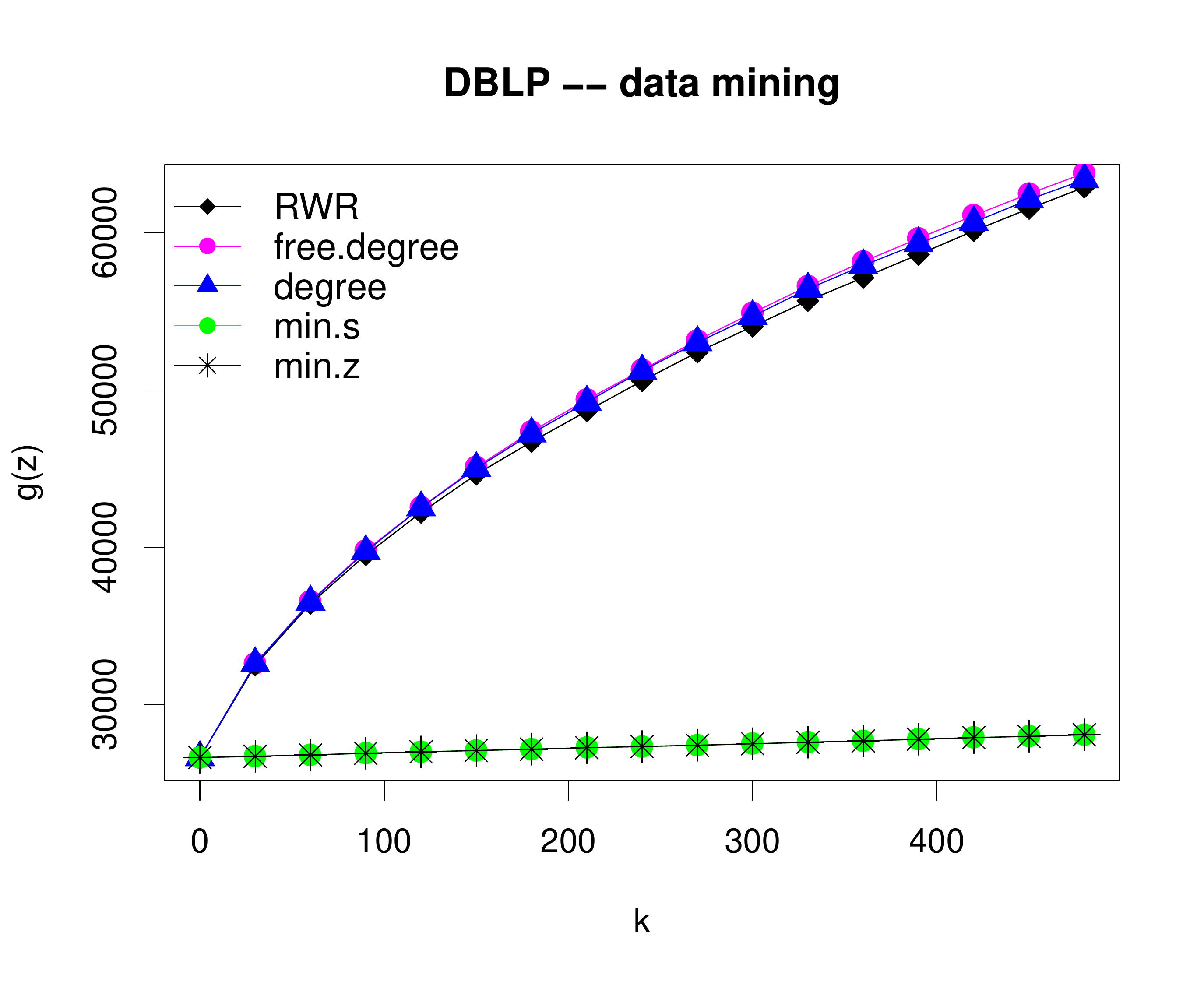}
\caption{\label{figure:bibsonomydblp}
Performance of the algorithms on the \bib\  and \dblp\  networks for the topic ``data mining''.}
\end{figure*}

The objective of our experiments is to
compare the proposed heuristics against \greedy, the algorithm with the approximation guarantee,  and demonstrate their scalability.

\subsection{Small networks.}

We experiment with a number of publicly available small networks.\footnote{www-personal.umich.edu/\~{}mejn/netdata/}
We evaluate our algorithms by reporting the value of
the objective function $g(\bfz)$ as a function of the solution set size $|T|=k$.
Our results for three small networks are shown in Figure~\ref{figure:newman}.
The datasets shown in the figure are the following:
($i$) \karate: a social network of friendships between 34 members of a karate club at a US university in the 1970s~\cite{zachary};
($ii$) \lesmis: co-appearance network of characters in the novel {\em Les Miserables}~\cite{graphbase}; and
($iii$) \dolphins: an undirected social network of frequent associations between 62 dolphins~\cite{dolphins}.
For this set of experiments we set the internal opinions $s_i$ to be a uniformly-sampled value in $[0,1]$.

We see that the difference between all the algorithms is relatively small but their relative performance is consistent.
The \greedy\ algorithm achieves the best results, while the three heuristics, \degree, \freedegree, and \rwr\ come close together.
{\mins} and {\minz} have the poorest performance, even though for larger values of $k$ they improve  and slightly outperform some of the heuristics.
Between the two, {\mins} performs best, outperforming {\minz}, especially for small values of $k$.
Both of those trends are expected: the three heuristics \degree, \freedegree, and \rwr, are better motivated than {\mins}, which in turn, is better motivated than {\minz}.

We obtain similar results for other small networks, although we do not provide the plots for lack of space.

\subsection{Bibliographic datasets: A ``data mining'' campaign.}

We also evaluate our algorithms on two large social networks, derived from bibliographic data.

The first dataset, \bib, is extracted from bibsonomy~\cite{bibsonomy}, a social-bookmarking and publication-sharing system.
From the available data, we extract a social graph of $45\,329$ nodes representing authors and $149\,895$ edges representing co-authorship relations.
For each author we also keep the set of tags that have been used for the papers of that author.

The second dataset, \dblp, is also a co-authorship graph among computer scientists extracted from the {\sc dblp} site.\footnote{www.informatik.uni-trier.de/\~{}ley/db/}
The dataset is a large graph containing $635\,585$ nodes and $1\,423\,716$ edges.
Again, for each author we keep the set of terms
they have been used in the titles of the papers they have co-authored.

In this experiment, we generate the internal opinion vectors
by identifying
keywords related to data mining (e.g., we picked
{\em data}, {\em mining}, {\em social}, {\em networks}, {\em graph}, {\em clustering},
{\em learning}, and {\em community}).
For each author $i\in V$ we then set $s_i$ to be the fraction of the above keywords present in his set of terms.
This setting corresponds to a hypothetical scenario of designing a campaign to promote the ``data mining'' topic among all computer-science researchers.

The results of the heuristics for the two datasets are shown in Figure~\ref{figure:bibsonomydblp}.
For the \bib\ dataset (left), there is a
clear distinction between the three
heuristics; {\rwr} clearly performs better
than both {\degree} and {\freedegree}.
This superior performance
of  \rwr\ is expected as this algorithm
takes into account both the degrees and
the values of the internal opinions~$s_i$.
Also the better performance of \freedegree\
compared with the performance of \degree\ is
consistent with the results obtained for smaller networks.
On the other hand, on the \dblp\ dataset (right part of Figure~\ref{figure:bibsonomydblp}) the behavior of
the three heuristics is more surprisingly,
as all three perform almost identical.
Finally, for both datasets, the difference of the three best heuristics \degree, \freedegree, and \rwr\ with the other two heuristics, {\mins} and {\minz} is more pronounced.
In fact, the performance of {\mins} and {\minz} is very poor.

We investigate the difference on the relative performance
of the best three heuristics, \degree, \freedegree, and \rwr, on the two datasets by plotting
the degrees of the nodes versus their $s_i$ value.
This is shown in Figure~\ref{figure:degree-vs-si}.
Recall that our intuition for selecting nodes that is to choose nodes that have large degree and small value of~$s_i$.
Figure~\ref{figure:degree-vs-si}
demonstrates that in the {\dblp} dataset, large degree correlates well
with small $s_i$ values, while this is not the case for the {\bib} dataset.
Therefore, for {\dblp} all three heuristics pick high-degree nodes, which makes their performance almost identical.

\begin{figure}[t]
\centering
\includegraphics[width=0.23\textwidth]{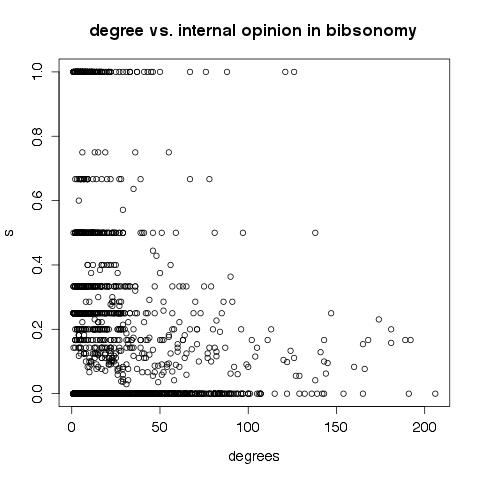}
\includegraphics[width=0.23\textwidth]{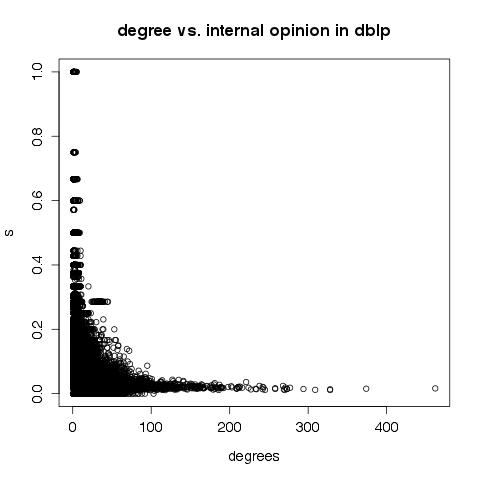}
\caption{\label{figure:degree-vs-si}
Scatter plot of degrees vs.\ internal opinion values $s_i$ in the two datasets, \bib\ and \dblp.}
\end{figure}

\section{Problem variants}
\label{section:variants}

The {\expressed} problem we studied in this paper focuses on campaigns that aim to alter the expressed opinions of individuals.
However, other campaign strategies are also possible. For example one could aim at altering the internal opinions of individuals such that the overall opinion is improved as much as possible.
Formally, the goal would be to select a set of nodes $S$, which are going to be convinced to change their {\em internal opinions} $s_i$ to $1$, such that  the resulting overall opinion~$g(\bfz\mid S)$ is maximized. We call this problem the {\internal} problem. The difference between the  {\expressed} and {\internal} problems is that in the former we are asking to fix the expressed opinions $z_i$ for $k$ individuals, while in the latter we
are asking to fix the internal opinions $s_i$. Even though the difference is seemingly small, the problems are 
computationally very different.


Algorithmically, the problem of selecting $S$ individuals to change their internal opinions, so that we maximize $g(\bfz\mid S)$ is much simpler. In fact, we can
show that for an undirected social graph $G=(V,E)$, where each node $i \in V$ has internal opinion $s_i$ and expressed opinion $z_i$, the following invariant holds, independently of the structure of the graph (the set of edges $E$):
\begin{equation}
\label{equation:invariant}	
g(\bfz) =\sum_{i \in V} z_i = \sum_{i\in V} s_i.
\end{equation}
Consequently, the goal of maximizing $g(\bfz)$ by modifying $k$ values $s_i$ can be simply achieved by selecting the $k$ smallest values $s_i$ and setting them to~1.
We note that the above observation does not hold once the expressed opinions of some individuals are fixed, as is the case in the {\expressed} problem.
The proof of the invariant, and its implications are discussed in the Appendix~\ref{appendix:invariant}.

The graph invariant has obvious implications for the other variant of the campaign problem, where we seek to maximize $g(\bfz)$ by adding or removing edges to the graph. From Equation~(\ref{equation:invariant}) it follows that for undirected social graphs this problem variant is meaningless; the overall opinion $g(\bfz)$ does not depend on the structure of the graph.
This observation has important implications for opinion formation on social networks. It shows that although the network structure has an effect on the individual opinions of network participants, it does not affect the average opinion in the network. For the campaign problem, this says that you cannot create more goodwill by altering the network. For the study of social dynamics, this implies that the collective wisdom of the crowd remains unaffected by the social connections between individuals.

\section{Conclusions}
\label{section:conclusions}

We considered a setting where opinions of individuals
in a social network evolve through processes of social dynamics, reaching a
Nash equilibrium. Adopting a standard social and economic model of such
dynamics
we addressed the following natural question: 
given a social network of
individuals who have their own internal opinions about an information
item, which are the individuals that need to be convinced to adopt a positive opinion
so that in the equilibrium state, the network (as a whole) has the maximum
positive opinion about the item?
We studied the computational complexity of this problem and
proposed algorithms for solving them exactly or approximately.
Our theoretical analysis and the algorithm design relied on  a connection
between opinion dynamics and random walks with absorbing states.
Our experimental evaluation on real datasets
demonstrated the efficacy of our
algorithms and the effect of the structural characteristics of the
underlying social networks on their performance.

\newpage

\appendix

\section{Proof of Theorem~\ref{theorem:NP}}
\label{appendix:np-hard}

\begin{proof}
We  prove the theorem by reducing an instance of the {\sc Vertex Cover on Regular Graphs} problem (VCRG)~\cite{feige03vertex}
to an instance of the decision version of the {\expressed} problem.
We remind that a graph is called regular if all its nodes have the same degree.

Given a regular graph $G_{VC}=(V_{VC},E_{VC})$ and an integer $K$ the VCRG
problem asks whether there exists a set of nodes $Y\subseteq V_{VC}$ such
that $|Y|\leq K$ and $Y$ is a vertex cover (i.e., for every $(i,j)\in E_{VC}$ it is $i\in Y$ or~$j\in Y$).

An instance of the decision version of the {\expressed} problem
consists of a social graph $G=(V,E)$, internal opinions $\bfs$, an integer $k$ and a number~$\theta$.
The solution to the decision version is ``yes'' iff there exists
a set $T\subseteq V$ such that $|T|\leq k$ and $\eg(\bfz\mid T)\geq \theta$.

Given an instance of the VCRG problem, we will construct an instance of the decision version of {\expressed}
by setting $G=(V,E)$ to be equal to $G_{VC}=(V_{VC},E_{VC})$, $s_i=0$ for every
$i\in V$, $k=K$ and $\theta = (n-k)\frac{d}{d+1}$. Then, we show
that $T\subseteq V$ is a solution to the {\expressed} problem with value $g(\bfz\mid T)\geq \theta$
if and only if $T$ is a vertex cover of the input instance of VCRG.

In order to show this, we use
the absorbing random walk interpretation of the {\expressed} problem.
Recall that by the definition of the {\expressed} problem, every node $i\in T$ becomes an absorbing node with value $z_i=1$.
For every other node $j\notin T$ we compute a value $z_j$, which is the expected value at the absorption point of a random walk that starts from $j$. Since $s_i = 0$ for all $i \in V$ and $z_i = 1$ for all $i \in T$,
the value $z_j$ represents the probability that the random walk starting
from $j$ will be absorbed in some node in $T$.

Now suppose that $T$ is a vertex cover for $G_{VC}$. Then, for every non-absorbing vertex $j \in V\setminus T$,
for each one of the $d$ edges $(j,i)\in E$ incident on $j$,
it must be that $i \in T$; otherwise edge $(j,i)$ is not covered, and $T$ is not a vertex cover.
Therefore, in the augmented graph $H$, node $j$ is connected to $d$ nodes in $T$, and to node $\sigma_j$, all of them absorbing.
A random walk starting from $j$ will be absorbed and converge in a single step.
The probability of it being absorbed in a node in $T$ is $z_j = \frac{d}{d+1}$.
There are $n-k$ nodes in $V\setminus T$, therefore, $\eg(\bfz \mid T) = (n-k)\frac{d}{d+1}$.


If $T$ is not a vertex cover for $G_{VC}$, then there is an edge $(j,\ell) \in E$, such that $j,\ell \not\in T$. As we noted before, $z_j$ is the probability of being absorbed in some node in $T$.
The transition probability of the edge $(j,\sigma_j)$ is $\frac 1{d+1}$, therefore, node $j$ has probability at least $\frac 1{d+1}$ of being absorbed in $\sigma_j$. Since there is a path from $j$ to $\sigma_\ell$ with non-zero probability, the probability of $j$ being absorbed in node $\sigma_\ell$ is strictly greater than zero. Therefore, the probability of being absorbed in some node not in $T$ is strictly greater than $\frac 1{d+1}$, and thus $z_j < \frac d {d+1}$. It follows that $\eg(\bfz \mid T) <(n-k)\frac{d}{d+1}$.
\end{proof}

\section{Proof of Lemma~\ref{theorem:submodular}}
\label{appendix:submodularity}

\begin{proof}
Recall that $\eg(\bfz\mid T) = \sum_{i\in V} z_i$, where the values of vector $\bfz$ are computed using Equation~\eqref{equation:equilibrium}.
As we have already described, we can view the computation of
the vector $\bfz$ as performing
a random walk with absorbing nodes on the augmented graph $H=(V\cup{\copynodes},E\cup {\copyedges})$.
Assume that $B$ is the set of absorbing nodes, and that
each $b\in B$ is associated with value $f_b$.
Let $P_B(b\mid i)$ be the probability that a random walk
that starts from $i$ gets absorbed at node $b$, when the set of absorbing nodes is $B$.
The expressed opinion of node $i\notin B$ is
\[
z_i = \sum_{b \in B}P_B(b\mid i)f_b.
\]
Since this value depends on the set $B$,
we will write $z(i\mid B)$ to denote the
value of $z_i$ when the set of absorbing nodes is $B$.

Initially, the set of absorbing nodes is $B = {\copynodes}$ and $f_b = s_b$ for all $b \in B$.
When we select a subset of nodes $T \subseteq V$ such that their expressed
opinions are fixed to 1, we have that $B = {\copynodes} \cup T$,
$f_b = s_b$ for all $b \in {\copynodes}$ and $f_b = 1$ for all $b \in T$.
Since the set of nodes ${\copynodes}$ is
always part of the set $B$, and the parameter that we are interested in for this proof is the
target set of nodes $T$, we will use $z(i\mid T)$ to denote $z(i\mid B)$ where $B = {\copynodes} \cup T$.
We thus have
\[
\eg(\bfz\mid T) = \sum_{i\in V} z(i \mid T).
\]
Note that the summation is over all nodes in $V$ including the nodes in $T$. If $i \in T$, then $P_B(i\mid i) = 1$, and $P_B(j\mid i) = 0$ for all $j \neq i$. Therefore, $z(i\mid T) = 1$ for all $i \in T$.

We now make the following key observation. Let $j$ be a non-absorbing node in $V\setminus T$. We have that
\[
z(i\mid T) = \sum_{b \in B} \left (P_{B\cup\{j\}}(b\mid i) + P_{B\cup\{j\}}(j\mid i)P_B(b\mid j)\right )f_b.
\]
The equation above follows from the observation that we can express the probability of a random walk starting from $i$ to be absorbed in some node $b \in B$ as the sum of two terms: (i) the probability $P_{B\cup\{j\}}(b\mid i)$ that the random walk is absorbed in $b$, while avoiding passing through $j$ (thus we add $j$ in the absorbing set); (ii) the probability $P_{B\cup\{j\}}(j\mid i)$ that the random walk is absorbed in $j$ while avoiding the nodes in $B$ (that is, the probability of all paths that go from $i$ to $j$ of arbitrary length, without passing through $j$ or $B$), times the probability $P_B(b\mid j)$ of starting a new random walk from $j$ and getting absorbed in $b$ (being able to revisit $j$ and any node in $V\setminus T$).

If we add node $j$ into the set $T$ we have that
\begin{eqnarray*}
\lefteqn{z(i\mid T\cup\{j\}) - z(i\mid T)}\\
& = & \sum_{b \in B\cup\{j\}} P_{B\cup\{j\}}(b\mid i)f_b \\
& & - \sum_{b \in B} \left (P_{B\cup\{j\}}(b\mid i) + P_{B\cup\{j\}}(j\mid i)P_B(b\mid j)\right )f_b \\
& = & P_{B\cup\{j\}}(j\mid i)\left( 1- \sum_{b \in B} P_B(b\mid j)f_b\right)\\
& = & P_{B\cup\{j\}}(j\mid i)\left(1-z(j\mid T)\right) \geq 0.
\end{eqnarray*}

Hence,
\[
\Delta \eg(T,j) = \left(1-z(j\mid T)\right)  \sum_{i \in V} P_{B\cup\{j\}}(j\mid i) \geq  0.
\]
Therefore, we can conclude that function
$\eg(\bfz\mid T)$ is monotone with respect to the set of target nodes $T$.

We now need to show that $\eg(\bfz\mid T)$ is submodular, that is for any $T,T'$ such that $T\subseteq T'$, and for any node $j \in V$, we have that $\Delta \eg(T',j) - \Delta \eg(T,j) \leq 0$. For
this we will use two random walks: one with absorbing
states $B={\copynodes} \cup T$ and the other
with absorbing states $B'= {\copynodes} \cup T'$. Following reasoning and notation similar to the one we used for monotonicity we have that
\begin{eqnarray*}
\lefteqn{\Delta \eg(T',j) - \Delta \eg(T,j)}\\
& = & \left(1-z(j\mid T')\right)\sum_{i \in V} P_{B'\cup\{j\}} (j\mid i) \\
& &- \left(1-z(j\mid T)\right)\sum_{i \in V} P_{B\cup\{j\}}(j\mid i) .
\end{eqnarray*}
From the monotonicity property we have that
\[
1-z(j\mid T') \leq 1-z(j\mid T).
\]
Also, as the number of
absorbing nodes increases, the probability of
being absorbed in a specific node $j$ decreases, since the probability of the random walk to be absorbed in a node other than $j$ increases.
Therefore, we also have that
\[
P_{B'\cup\{j\}} (j\mid i) \leq P_{B\cup\{j\}}(j\mid i).
\]
Combining these last two observations we conclude that
$\Delta \eg(T',j) - \Delta \eg(T,j) \leq 0$ which shows that the function
is submodular.
\end{proof}

\section{Graph invariants}
\label{appendix:invariant}

In this section we prove a graph invariant related to the sum of the values $z_i$ and~$s_i$.
This invariant has repercussions in the following scenarios:
\begin{itemize}
\item[($i$)] maximize $g(\bfz)$ by modifying only the internal opinions $s_i$ of the users (problem {\internal} in Section~\ref{section:variants}); and
\item[($ii$)] maximize $g(\bfz)$ by adding edges in the social graph, for instance, recommend friendships or certain accounts for users to connect and follow.
\end{itemize}

We prove the invariant in a slightly more general setting than the one we consider in the paper.
We then formulate the more special case of the invariant for our exact problem setting, and we discuss its implications in the above-mentioned scenarios ($i$) and~($ii$).

Consider an undirected graph $G = (V,E)$.
We use $w_{ij}$ to denote the weight of edge $(i,j)$ and $W_i$ to denote the total weight of all edges incident on node $i$.
We assume that the vertices in $V$ are partitioned in two sets $U$ and $B$.
The nodes in $B$ are \emph{absorbing} nodes for the random walk.
Each node $j \in B$ is associated with a value $f_j$. The value $f_j$ can be either the internal opinion of node $j$ (in which case $f_j = s_j$), or the fixed expressed opinion of node~$j$ (in which case $f_j=1$).
For each node $i \in U$ we will compute a value $z_i$ which is the expected value at the node of absorption for a random walk that starts from node $i$, as given by Equation~(\ref{equation:equilibrium}).

We further make the assumption that the set of absorbing nodes can be partitioned into $|U|$ disjoint subsets $\{B(i)\}$,
one for each node $i \in U$, such that the nodes in $B(i)$ are connected \emph{only} with the node $i$.
We can make this assumption without loss of generality, since in the case that a node $j \in B$ is connected to $k$ nodes $\{i_1,...,i_k\}$ in $U$, we can create $k$ copies of $j$, each with value $f_j$, and connect each copy with a single node in $U$ with an edge of the same weight, while removing the original node $j$ from the graph.
In the resulting graph the $z_i$ values computed by the absorbing random walk are the same as in the original graph.

To introduce some additional notation let $E_U$ denote the set of edges between non-absorbing nodes, and let $E_B$ denote the set of edges between nodes in $U$ and in $B$.
Note that by the construction above there are no edges between the nodes in $B$. Such edges would not have any effect anyway, since the nodes in $B$ are absorbing. Given a node $i \in U$, let $N(i)$ denote the set of neighbors of $i$, let $U(i)$ denote the set of non-absorbing neighbors of $i$, and let $B(i)$ denote the set of absorbing neighbors of $i$.

From the definition of the absorbing random walk we have that
\[
z_i = \frac{1}{W_i}\sum_{j:j\in B(i)} w_{ij}f_j + \frac{1}{W_i}\sum_{j:j\in U(i)} w_{ij}z_j,
\]
and thus
\[
W_i z_i = \sum_{j:j\in B(i)} w_{ij}f_j + \sum_{j:j\in U(i)} w_{ij}z_j.
\]
Summing over all $i \in U$ we get
\begin{eqnarray}
\sum_{i\in U} W_i z_i & = & \sum_{i\in U}\sum_{j:j\in B(i)} w_{ij}f_j + \sum_{i\in U}\sum_{j:j\in U(i)} w_{ij}z_j \nonumber \\
\label{equation:invariant_1}
& = & \sum_{(i,j) \in E_B} w_{ij}f_j + \sum_{(i,j) \in E_U} w_{ij}\left(z_j + z_i\right).
\end{eqnarray}
The left-hand side can also be written as:
\begin{eqnarray}
\sum_{i\in U} W_i z_i & = & \sum_{i\in U} \sum_{j\in N(i)} w_{ij} z_i \nonumber \\
 & = & \sum_{i\in U} \sum_{j\in U(i)} w_{ij} z_i + \sum_{i\in U} \sum_{j\in B(i)} w_{ij} z_i \nonumber \\
\label{equation:invariant_2}
 & = & \sum_{(i,j) \in E_U} w_{ij}\left(z_j + z_i\right) + \sum_{(i,j) \in E_B} w_{ij}z_i .
\end{eqnarray}
By Equations~(\ref{equation:invariant_1}) and~(\ref{equation:invariant_2}) we obtain
\begin{equation}
\label{eq:general:z-s}
\sum_{(i,j) \in E_B} w_{ij}\left(z_i - f_j\right) = 0.
\end{equation}
Equation~(\ref{eq:general:z-s}) is the most general form of our invariant.
The equation relates the values of $z_i$ and $f_j$ via the weights $w_{ij}$ across the edges $E_B$, i.e., only the edges between absorbing and non-absorbing nodes.
The set of edges $E_U$ between the non-absorbing nodes does not play any role.

We now consider the special case in which Equation~(\ref{eq:general:z-s}) is applied to graphs considered in this paper, that is, in augmented graphs of type $H = ( V\cup \copynodes, E\cup \copyedges)$, as defined in Section~\ref{section:notation}.
In that case, each node $i \in V$ is connected to a single absorbing node $\sigma(i) \in \copynodes$, which has value $f_{\sigma(i)} = s_i$.
Furthermore, for each edge $(i,j) \in \copyedges$, we have $w_{ij} = w>0$, namely, all edges to absorbing nodes have the same weight.
In this case, the invariant becomes
\begin{equation}
\tag{\ref{equation:invariant}}	
g(\bfz) =\sum_{i \in V} z_i = \sum_{i\in V} s_i.
\end{equation}
As already discussed in Section~\ref{section:variants},
Equation~(\ref{equation:invariant}) has the following implications.

\begin{itemize}
\item[($i$)]
Regarding the problem {\internal}, that is, when we ask to maximize $g(\bfz)$ by modifying only the internal opinion values $s_i$, it is easy to see that the maximum increase occurs when selecting the $k$ smallest values $s_i$ and setting them to~1.
This observation motivates the algorithm {\mins} described in Section~\ref{section:algorithms}.
\item[($ii$)]
Consider the following problem: We want to maximize $g(\bfz)$ by only adding or removing edges in the social graph and without modifying any of the values $s_i$ or $z_i$.
Equations~(\ref{eq:general:z-s}) and~(\ref{equation:invariant}) provide an expression for $g(\bfz)$ that is independent on the structure of the graph defined by the edges in $E_U$, and thus, show that it is not possible to change $g(\bfz)$ by adding or removing edges.
\end{itemize}

\flushend

\end{document}